%% file: main.tex
\documentclass[envcountsame]{llncs}
\usepackage{amssymb}
\usepackage{amsmath}
\usepackage{mathtools}
\usepackage{relsize}

\usepackage{framed}

\usepackage{url}
\usepackage{algpseudocode}
\usepackage{algorithm}

\usepackage{enumitem}
\usepackage{xcolor}
\usepackage{comment}
\usepackage{tikz}
\usepackage{epsfig}
\usetikzlibrary{shapes,backgrounds, patterns}

\colorlet{bscolor}{green}
\colorlet{skcolor}{orange}

\newcommand{\Omit}[1]{}

\newcommand{\cf}{CF }
\newcommand{\cfcn}{CF-CN }

\begin{document}

\title{Conflict-Free Coloring on Open Neighborhoods}
\author{
Sriram Bhyravarapu
\and
Subrahmanyam Kalyanasundaram
}
\institute{Department of Computer Science and Engineering, IIT Hyderabad \\
\email{\{cs16resch11001, subruk\}@iith.ac.in}
}

\maketitle

\begin{abstract}
In an undirected graph, a conflict-free coloring 
(with respect to open neighborhoods) 
is an assignment of colors to the vertices of the graph $G$ such that 
every vertex in $G$ has a uniquely colored vertex in its open 
neighborhood. 
The conflict-free coloring problem asks to find the smallest number of 
colors required for a conflict-free coloring.

The conflict-free coloring problem is NP-complete. From results in Abel et. al. [SODA 2017],
it can be inferred that every planar graph has a conflict-free coloring with at most nine 
colors.
As the best known lower bound for planar graphs is four colors, 
it was asked in the same paper if fewer colors would  suffice. 
We make progress in answering this question, by showing
that every planar graph can be colored using at most six colors.
The same proof idea is used to show that every outerplanar
graph can be colored using at most five colors. 
Using a different approach, we further show that every outerplanar 
graph can be colored using at most four colors. 

Finally, we study the problem on Kneser graphs. We show that $k+2$ colors are necessary and sufficient to color the Kneser graph 
$K(n,k)$ when $n\geq k(k+1)^2 + 1$.
\end{abstract}

\input{history}
\input{preliminaries}
\input{planar5colors.tex}
\input{completecoloring.tex}

\input{kneser.tex}

\input{conclusion}



\bibliographystyle{elsarticle-num}

\bibliography{BibFile}

\end{document}

%% file: history.tex
\section{Introduction}
A \emph{proper coloring} of a graph is an assignment of a color to every vertex of the graph such that adjacent vertices 
are of distinct colors.
Conflict-free coloring is a variant of the graph coloring problem. 
A conflict-free coloring of a graph $G$ is a coloring such that for every vertex in $G$, there exists a uniquely colored vertex in its neighborhood. 
This problem was first introduced in 2002 by Even, Lotker, Ron and Smorodinsky \cite{Even2002}. 
This problem was originally motivated by wireless communication systems consisting of base stations and clients. The clients and base stations have to send each other
information and hence they communicate with each other.
Each base station is assigned a frequency and if two base stations with 
the same frequency try to communicate with a client, it leads to interference. 
So for each client, there has to be a base station with a unique frequency. 
Since each frequency band is expensive, there is a need 
to minimize the number of frequencies used 
by the base stations. Over the past two decades, this problem has been very well studied, see for instance
the survey by Smorodinsky \cite{smorosurvey}.

The conflict-free coloring problem has been studied with respect to open  neighborhoods as well as closed neighborhoods.
In this paper, we study the conflict-free coloring problem with respect to open neighborhoods. 

\begin{definition}[Conflict-Free Coloring]\label{def:cf}
A \emph{complete conflict-free (CF) coloring} of a graph $G = (V,E)$ using $k$ colors is an 
assignment $C:V(G) \rightarrow \{1, 2, \ldots, k\}$ such that for every $v \in V(G)$, 
there exists an $i \in \{ 1, 2, \ldots, k\}$ such that  $|N(v) \cap C^{-1}(i)| = 1$.
The smallest number of colors required for a complete conflict-free coloring
of $G$ is called the conflict-free chromatic number of $G$, denoted by $\chi_{CF}(G)$.
\end{definition}

The conflict-free coloring problem and many of its variants are known to be 
NP-complete~\cite{planar,gargano2015}. It was further shown in \cite{gargano2015} that the \cf coloring problem on open 
neighborhoods is hard to approximate within a factor of $n^{1/2 - \varepsilon}$,
unless P = NP. 
Since the problem is NP-hard, the parameterized aspects of the problem 
have been studied. 
The problems are fixed parameter tractable when 
parameterized by vertex cover number, neighborhood 
diversity~\cite{gargano2015}, distance to cluster, distance to threshold graphs~\cite{vinod2017}, and more recently, tree-width~\cite{Boe2019,Sak2019}.

In this paper, we look at the conflict-free open neighborhood problem, which 
is considered as the harder of the open and closed neighborhood variants, see
for instance, remarks in \cite{Pach2009,Keller2018}. It is easy to construct examples
of bipartite graphs $G$, for which $\chi_{CF}(G)$ is $\Theta(\sqrt{n})$. Since any proper 
coloring is also a valid conflict-free closed neighborhood coloring, these examples
have a \cf closed neighborhood coloring using two colors. Further, Cheilaris 
\cite{pcheilaris2009} showed that for every graph $G$, we have $\chi_{CF}(G) \leq 2 \sqrt{n}$. 
On the contrary, a graph with maximum degree $\Delta$ has a conflict-free closed neighborhood
coloring with at most $O(\log^{2 + \varepsilon} \Delta)$ colors \cite{Pach2009}.

Restrictions of the conflict-free coloring problem to special classes of graphs 
have been studied extensively. Of these, graphs arising out of intersection of 
geometric objects have attracted special interest, see for instance, \cite{Keller2018,sandor2017,chen2005}. 
The problem has also been studied for structural classes of 
graphs such as bipartite graphs and split graphs~\cite{vinod2017}.

In~\cite{planar}, Abel et. al. considered the partial coloring variant of the problem where
not all vertices need to be assigned a color. 
\begin{definition}[Partial Conflict-Free Coloring]
A \emph{partial conflict-free (CF) coloring} of a graph $G= (V,E)$ using $k$ colors is an 
assignment $C:V(G) \rightarrow \{0, 1, 2, \ldots, k\}$ such that for every $v \in V(G)$, 
there exists an $i \in \{1, 2, \ldots, k\}$ such that  $|N(v) \cap C^{-1}(i)| = 1$.
\end{definition}
Let us refer to the variant of the problem that we originally stated in Definition 
\ref{def:cf}, where all
the vertices have to be colored, as the complete coloring variant. 
The key difference between partial \cf coloring and complete \cf coloring is that in the partial variant, we allow some
vertices to be assigned the color 0. It is convenient to think of the vertices 0 
as uncolored vertices. However, the uniquely colored neighbor is not allowed to be
of color 0. 
If a graph can be colored using $k$ colors in the 
partial coloring variant, then all the uncolored vertices can be assigned 
the color $k+1$, and thus a $k+1$ complete coloring for the same graph can be 
obtained.

For the partial coloring variant,
eight colors suffice to color a planar graph~\cite{planar}. 
It is easy to construct a planar graph that requires four colors. Starting with $K_4$,  each original edge is subdivided by introducing a 
degree-two vertex on this edge. In addition, a pendant vertex is attached to every original 
vertex of the $K_4$. This graph (see Figure \ref{fig:k4}) is planar and
requires four colors. 
One of the open questions asked in \cite{planar} was to close the gap between the upper bound 
of eight and lower bound of four for the partial coloring variant of the 
conflict-free chromatic number of a planar graph. 
In this paper 
we reduce this gap, by showing that five colors suffice for the partial \cf
coloring of a planar graph. Using the same proof idea, we show that four colors
are enough for the partial coloring of an outerplanar graph. The lower bound
for the partial \cf coloring of an outerplanar graph is three, see Figure \ref{fig:k3}. 
\input{fig.tex} 
As noted before, a partial coloring of an outerplanar graph using four colors implies
a complete coloring using five colors. Using a different approach, we show that 
 four colors are sufficient for a complete \cf coloring of an outerplanar graph. 

The last section in this paper studies the
conflict-free coloring on Kneser graphs. The Kneser graph $K(n,k)$ is the graph whose vertices are 
$k$-subsets of $[n]$, and two such vertices are adjacent if and only if the corresponding sets are disjoint. 
Several properties of Kneser graphs have been subject to study.
The chromatic number of the Kneser graph $K(n,k)$ was conjectured by Kneser \cite{kneser} in 1955 to be $n-2k+2$. This remained open till 
Lov\'asz proved \cite{lovasz} the conjecture in 1978. When $n \geq k(k+1)^2 + 1$, we determine the exact conflict-free chromatic number of the Kneser graph $K(n,k)$.


We summarize our results in this paper below:
\vspace{-0.1in}
\begin{enumerate}
    \item Five colors are sufficient 
    for the partial conflict-free 
    coloring of a planar graph. This improves the
    previous best known bound of~\cite{planar} that 
    required eight colors.
    
    Four colors are sufficient for the partial
    conflict-free coloring of an outerplanar graph. 
    These two results are
    discussed in Section \ref{sec:planar}.
    \item Four colors suffice for 
    the complete conflict-free 
    coloring of an outerplanar graph.
    Moreover, three colors are sufficient and sometimes necessary 
    for a complete conflict-free coloring of cactus graphs. 
    These results are shown in Section \ref{sec:outer}.
    \item In Section \ref{sec:kneser}, we compute bounds on the
    conflict-free coloring of Kneser graphs. We also determine that 
    the $\chi_{CF}(K(n,k)) = k + 2$ when $n \geq k(k+1)^2 + 1$.
\end{enumerate}{}

%% file: fig.tex
\begin{figure}[t!]
\centering
\begin{minipage}{0.4\textwidth}
\centering
\begin{tikzpicture}[every node/.style={node distance=1.5cm,scale=0.5}, scale = 0.8]
\tikzstyle{vertex}=[circle,draw, minimum size=3pt]
\tikzstyle{edge} = [draw,thick,-,black]
\node[vertex]  (u1) at (4,2) {};
\node[vertex]  (u2) at (6,2) {};
\node[vertex]  (u3) at (6,0) {};
\node[vertex]  (u4) at (4,0) {};
\node[vertex]  (u5) at (3,3) {};
\node[vertex]  (u6) at (7,3) {};
\node[vertex]  (u7) at (4.7,0.7) {};
\node[vertex]  (u8) at (3,-1) {};
\node[vertex]  (u9) at (5,2) {};
\node[vertex]  (u10) at (4,1) {};
\node[vertex]  (u11) at (6,1) {};
\node[vertex]  (u12) at (5,0) {};
\node[vertex]  (u13) at (5.3,1.3) {};
\node[vertex]  (u14) at (7,-1) {};
\draw[edge, color=black] (u1) -- (u5) (u1) -- (u9) (u1) -- (u10) (u1) -- (u13) ;
\draw[edge, color=black] (u2) -- (u9) (u2) -- (u11) (u2) -- (u6) (u2) -- (u14) ;
\draw[edge, color=black] (u3) -- (u7) (u3) -- (u11) (u3) -- (u13) (u3) -- (u12) ;
\draw[edge, color=black] (u4) -- (u12) (u4) -- (u10) (u4) -- (u8) (u4) -- (u14) ;
\end{tikzpicture}
\caption{Graph that requires 4 colors for a partial CF coloring.}\label{fig:k4}
\end{minipage}\hspace{0.1\textwidth}
\begin{minipage}{0.4\textwidth}
\centering
\begin{tikzpicture}[every node/.style={node distance=1.5cm,scale=0.5}, scale = 0.8]
\tikzstyle{vertex}=[circle,draw, minimum size=3pt]
\tikzstyle{edge} = [draw,thick,-,black]
\node[vertex]  (v1) at (0,2) {};
\node[vertex]  (v2) at (-1,0) {};
\node[vertex]  (v3) at (1,0) {};
\node[vertex]  (v4) at (-0.5,1) {};
\node[vertex]  (v5) at (0.5,1) {};
\node[vertex]  (v6) at (0,0) {};
\node[vertex]  (v7) at (0,3) {};
\node[vertex]  (v8) at (-2,-1) {};
\node[vertex]  (v9) at (2,-1) {};
\draw[edge, color=black] (v7) -- (v1) (v1) -- (v4);
\draw[edge, color=black] (v1) -- (v5) (v2) -- (v4);
\draw[edge, color=black] (v2) -- (v6) (v2) -- (v8);
\draw[edge, color=black] (v6) -- (v3) (v5) -- (v3) (v9) -- (v3);
\end{tikzpicture}
\caption{Graph that requires 3 colors for a partial CF coloring.}\label{fig:k3}
\end{minipage}
\vspace{-4mm}
\end{figure}

%% file: preliminaries.tex
\section{Preliminaries}\label{sec:prelim}
For any two vertices $u,v\in G$, we denote the shortest distance between them in $G$ by $dist(u,v)$. 
The open neighborhood of $v$, denoted by $N(v)$, is the set of vertices adjacent to $v$. 
We denote the graph induced by a set of vertices $V'$  in $G$ as $G[V']$. 

In this paper,
we consider only connected graphs with at least two vertices
because the colorings of the connected components can combine to
give a coloring of the graph. 
Also, an isolated vertex does not have a conflict-free coloring (in the open neighborhood setting)
since there are no neighbors.

A \emph{planar graph} is a graph that can be drawn in $\mathbb R^2$ (a plane) such 
that the
edges do not cross each other in the drawing.
Each such drawing divides the plane into regions and each region is called a 
\emph{face}. A planar drawing of a graph has one face that is unbounded. This face is called the \emph{outer face}. All the other faces are referred to as
\emph{inner faces}.
An \emph{outerplanar graph} is a planar graph that has a drawing in a plane such 
that all the vertices of the graph belong to the outer face. 
Throughout the paper, we use terminology from the textbook ``Graph Theory'' by Diestel \cite{Diestel}.

%% file: planar5colors.tex
\section{Partial CF Coloring of Planar Graphs}\label{sec:planar}
In~\cite{planar} Abel et. al., showed that eight colors are sufficient for the
partial \cf 
coloring of a planar graph. 
In this section, we improve the bound to five colors. 


We need the following definition:
\begin{definition}[Maximal Distance-3 Set]
For a graph $G= (V,E)$, a \emph{maximal distance-3 set} is a set $S \subseteq V(G)$
that satisfies the following:
\begin{enumerate}
    \item For every pair of vertices $w, w' \in S$, we have $dist(w,w') \geq 3$.
    \item For every vertex $w \in S$,  $\exists w' \in S$ such that
            $dist(w, w') = 3$.
    \item For every vertex $x \notin S$, $\exists x' \in S$ such 
            that $dist(x, x') < 3$.
\end{enumerate}
\end{definition}

The set $S$ is constructed by initializing $S=\{v\}$ where $v$ is 
an arbitrary vertex. We proceed in iterations. In each iteration, we add 
a vertex $w$ to $S$ if (1) for every $v$ already in $S$, $dist(v,w) \geq 3$,
and (2) there exists a vertex $w'\in S$ such that $dist(w,w') = 3$.
We repeat this 
until no more vertices can be added. 


The main component of the proof is the construction of an auxiliary graph 
$G'$ from the given graph $G$. 

\vspace{1mm}
\noindent\textbf{Construction of $G'$:} 
The first step is to pick a maximal distance-3 set $V_0$.
Notice that any distance-3 set is an independent set by definition.
We let $V_1$ denote the neighborhood of $V_0$. More formally,
$V_1 = \{ w\mid \{w, w'\} \in E(G), w' \in V_0\}$.
Let $V_2$ denote the remaining vertices i.e., 
$V_2=V\setminus (V_0 \cup V_1)$. 

We note the following properties satisfied by the above partitioning of $V(G)$.
\begin{enumerate}
    \item The set $V_0$ is an independent set.
    \item For every vertex $w \in V_1$, there exists a unique vertex $w' \in V_0$ such that $\{w, w'\} \in E(G)$. This is because if there are two such vertices, this will violate
    the distance-3 property of $V_0$.
    \item Every vertex in $V_0$ has a neighbor in $V_1$. If there exists $v \in V_0$
            without a neighbor in $V_1$, then $v$ is an isolated vertex. By assumption,
            $G$ does not have isolated vertices.
    \item There are no edges from $V_0$ to $V_2$.
    \item Every vertex in $V_2$ has a neighbor in $V_1$, and is hence at distance 2 from some vertex in $V_0$. This is due to the maximality of the distance-3 set $V_0$. 
\end{enumerate}

Now we define $A=V_0\cup V_2$. 
We first remove all the edges of $G[V_2]$ making $A$ an independent set.
We then contract every vertex $v\in A$ to a neighbor $f(v) \in N(v) \subseteq V_1$. 
The contraction process is as follows: we first identify vertex $v$ with $f(v)$. 
Then for every edge $\{v, v'\}$, we add an edge $\{f(v),v'\}$.
The resulting graph is $G'$. 


\begin{theorem}\label{graphon}
Five colors are sufficient to partial \cf 
color a planar graph.
\end{theorem}
\begin{proof}
Let $G$ be a planar graph. 
We first construct the graph $G'$ as above.
Since the steps for constructing $G'$ involve only edge deletion, 
and edge contraction, $G'$ is also a planar graph.
By the planar four-color theorem \cite{4ct}, every planar graph
has a proper four coloring. That is, 
there is an assignment $C: V(G') \rightarrow \{2,3,4,5\}$
such that no two adjacent vertices of $G'$ are assigned the same color.
Notice that we have colored all the vertices of $G'$, that is the entire set $V_1$.

Now, we extend $C$ to get a \cf 
coloring for $G$. 
For all vertices $v \in V_0$, we assign $C(v) = 1$.
The vertices in $V_2$ are assigned the color 0.

We will show that $C$ is indeed a partial \cf 
coloring of $G$. 
Consider a vertex $v\in A$ which is contracted to a neighbor $f(v)=w\in V_1$. 
The color assigned to $w$ is distinct from all $w$'s neighbors
in $G'$. 
Hence the color assigned to $w$ is the unique color 
among the neighbors of $v$ in $G$.

For each vertex $w \in V_1$, $w$ is a neighbor of exactly one vertex $v\in V_0$.
Every vertex $v \in V_0$ is colored 1, which is different 
from all the colors assigned to the neighbors of $w$ in $G'$.
\qed
\end{proof}

\noindent \textbf{Algorithmic Note:} The steps in 
the proof of Theorem \ref{graphon} lead to an algorithm.
The steps involved are construction of maximal distance-3 set, contraction 
of vertices in $A$ and the planar 4 coloring \cite{4ct}. All these steps
can be performed in $O(|V(G)|^2)$ time. Thus we have an $O(|V|^2)$ time algorithm,
that given a planar graph $G$, determines a partial \cf coloring for $G$ that uses
five colors.

Outerplanar graphs have a proper coloring using three colors. By 
argument analogous to Theorem \ref{graphon}, we infer the following.

\begin{corollary}\label{cor:outer}
Four colors are sufficient to partial \cf 
color an outerplanar graph.
\end{corollary}

The famous Hadwiger's conjecture states that if a graph $G$ does not 
contain $K_{k+1}$ as a minor, then $G$ is $k$-colorable (in the sense
of proper coloring). By an analogous argument again, we obtain the following:

\begin{corollary}\label{hadwiger}
Suppose the Hadwiger's conjecture is true and that G has no $K_{k+1}$ minor. Then 
$G$ admits a partial \cf 
coloring using $k+1$ colors.
\end{corollary}

%% file: completecoloring.tex
\section{Complete CF Coloring of Outerplanar Graphs}\label{sec:outer}

We saw in Corollary \ref{cor:outer} that outerplanar graphs can be partially \cf
colored using 4 colors. 
This implies a complete \cf coloring using 5 colors. In this section, we show 
an improved bound.

\begin{theorem}\label{thm:4col}
Four colors are sufficient for a complete 
\cf coloring of an outerplanar graph $G$.
\end{theorem}
Note that whenever we refer to an outerplanar
graph $G$, we will also be implicitly referring
to a planar drawing of $G$ with all the vertices appearing in the outer face. We will 
abuse language and say ``faces of $G$'' when 
we want to refer to faces of the above planar 
drawing. 

Theorem \ref{thm:4col} is proved using a two-level induction process. The first level
is using a \emph{block decomposition} of the graph. Any connected graph can be viewed
as a tree of its constituent blocks. We color the blocks in order so that when we
color a block, at most one of its vertices is previously colored. Each block is colored 
without affecting the color of the already colored vertex. The second level 
of the induction is required for coloring each of the blocks. We use
ear decomposition on each block and color the faces of the block in sequence.
However, the proof is quite technical and involves several cases of analysis at each 
step.

We summarize the relevant aspects of block decomposition
below.
The reader is referred to a standard textbook in graph theory \cite{Diestel} 
for more details on this.
\vspace{-0.05in}
\begin{itemize}
    \item A \emph{block} is a maximal connected subgraph without a cut vertex.
    \item Blocks of a connected graph are either maximal 2-connected subgraphs, or edges (the edges which form a block will be bridges).
    \item Two distinct blocks overlap in at most one vertex, which is a cut vertex.
    \item Any connected graph can be viewed as tree of its constituent blocks. 
\end{itemize}
In the following discussion, we explain how to construct a coloring 
$C: V(G) \rightarrow \{1, 2, 3, 4\}$ for an outerplanar graph $G$. At any 
intermediate stage, the coloring
$C$ will satisfy\footnote{The condition marked $\star$ is violated in a few cases. In the 
exceptional cases where it is violated, we shall explain how the cases are handled.} 
the following invariants:
\begin{framed}
\noindent \textbf{Invariants of $C$}
\begin{itemize}
    \item  Every vertex $v$ that has already been assigned a color $C(v)$ has a neighbor $w$, such that $C(w) \neq C(x)$, where
    $x \in N(v) \setminus \{w\}$. 
    For $v$, the function $U: V(G) \rightarrow \{1, 2, 3, 4\}$ denotes 
    the color of $w$, its uniquely colored neighbor.
    
    \item $\forall v\in V(G)$, $C(v)\neq U(v)$.
    \item $\forall \{v, w\} \in E(G)$, $C(v) \neq C(w)$
     and $|\{C(v),U(v),C(w),U(w)\}|=3$. 
     \textbf{($\star$)}
\end{itemize}
\end{framed}

Theorem \ref{thm:4col} is proved by using an induction
on the block decomposition of the graph $G$ and 
the below results. 

\begin{lemma}\label{lem:5cycles}
If $G$ is a 2-connected outerplanar graph such that all its inner faces contain exactly 5 vertices, then $G$ has a complete \cf 
coloring 
using 3 colors.
\end{lemma}

\begin{theorem}\label{thm:otherthan5cycle}
Let $G$ be an outerplanar graph.
\begin{enumerate}
    \item If $B$ is a block of $G$ that is either a bridge, or contains an inner 
face $F$ with $|V(F)|\neq 5$, then $B$ has a complete \cf 
coloring
using at most 4 colors.

    \item If $B$ is a block of $G$, with exactly one vertex $v$ 
precolored with color $C(v)$ and unique color $U(v)$,
then the rest of $B$ has a complete \cf 
coloring
using at most 4 colors, while retaining $C(v)$ and 
$U(v)$. 
\end{enumerate}
\end{theorem}

\begin{proof}[Proof of Theorem \ref{thm:4col}]
Let $G$ be an outerplanar graph. 
We apply block decomposition on $G$ which results in 
blocks that are either maximal 2-connected subgraphs or single edges. 

If $G$ is 2-connected and all its inner 
faces have exactly 5 vertices, then by Lemma 
\ref{lem:5cycles}, $G$ has a complete \cf coloring using 3 colors.

If $G$ does not fit the above description, 
then $G$ has a block $B$ such that either $B$ is
an edge, or $B$ has an inner face $F$ with 
$|V(F)|\neq 5$. In this case, by Theorem
\ref{thm:otherthan5cycle}.1, $B$ has a complete 
\cf coloring using at most 4 colors.

Viewing $G$ as a tree of its blocks, we can 
start coloring blocks that are adjacent to 
blocks that are already colored. Suppose the
block $B$ is already colored, and let $B'$ be a
block adjacent to $B$. Let $x$ be the cut-vertex between
the blocks $B$ and $B'$.
We use Theorem
\ref{thm:otherthan5cycle}.2 to obtain a complete 
\cf coloring of $B'$ using at most 4 colors, while retaining $C(x)$ and $U(x)$.
\qed
\end{proof}

We now proceed towards proving Lemma \ref{lem:5cycles} and Theorem \ref{thm:otherthan5cycle}.
Lemma \ref{lem:5cycles} and Theorem \ref{thm:otherthan5cycle} 
discusses the coloring of blocks, which is accomplished by means of induction on the faces of the blocks. Towards this end, 
we use the following fact about ear decomposition of 2-connected outerplanar graphs. 
For a proof of the below lemma, we refer the reader to 
\cite{Aubry2016} where this is stated as Observation 2.

\begin{lemma}[Ear Decomposition]\label{lem:ear}
Let $B$ be a 2-connected block in an outerplanar graph. Then $B$ has an ear decomposition $F_0, P_1, P_2, \ldots, P_q$ satisfying the following:
\begin{itemize}
    \item $F_0$ is an arbitrarily chosen inner face of $B$.
    \item Every $P_i$ is a path with end points $v, w$
    such that $\{v, w\}$ is an edge in $F_0 \cup \bigcup_{1 \leq j<i} P_j$.
    Thus $P_i$ together with the edge $\{v, w\}$ forms a face of $B$.
\end{itemize}
\end{lemma}
We first prove Lemma \ref{lem:5cycles}.
\begin{proof}(Proof of Lemma \ref{lem:5cycles})
Since $G$ is 2-connected, the entire graph forms a single block. 
Let $F_0, P_1,\dots, P_q$ be an ear decomposition of $G$. 
Recall that all the faces have exactly five vertices.
Let $F_0 = v_1-v_2-v_3-v_4-v_5-v_1$.
We assign\footnote{The coloring assigned in this proof
does not satisfy the condition marked $\star$. However,
this is not an issue since we are coloring the whole of
$G$ in this lemma.}
the following colors to the vertices in $F_0$: $C(v_1) = 1, C(v_2) = 1, C(v_3) = 2, C(v_4) = 2, C(v_5) = 3$. 
We also have $U(v_1) = 3, U(v_2) = 2, U(v_3) = 1, U(v_4) = 3, U(v_5) = 1$.

Let $P_i$ be any subsequent face $P_i = w_1-w_2-w_3-w_4-w_5-w_1$ with 
$\{w_1, w_2\}$ being the pre-existing edge in $F_0 \cup \bigcup_{1 \leq j<i} P_j$. Depending on the values already assigned to 
$C(w_1), U(w_1), C(w_2), U(w_2)$, we assign the colors to $w_3, w_4$ and $w_5$. 
We always ensure that $C(v) \neq U(v)$ for all vertices $v$. 
We note that the values $C(w_1), U(w_1), C(w_2), U(w_2)$ can take only 
the four below combinations, w.l.o.g. 
We explain the coloring for the rest of $P_i$ in each of these cases.
\begin{enumerate}
    \item $C(w_1) = C(w_2)$ and $|\{C(w_1), U(w_1), U(w_2)\}| = 3$. W.l.o.g., let $C(w_1) = 1, U(w_1) = 2, C(w_2) = 1, U(w_2) = 3$. Assign $C(w_3) = 2, C(w_4) = 2, C(w_5) = 3$ and $U(w_3) = 1, U(w_4) = 3, U(w_5) = 1$.
    \item $C(w_1) \neq C(w_2)$, $U(w_1) \neq U(w_2)$, and $|\{C(w_1), C(w_2), U(w_1), U(w_2)\}| = 3$. Either $w_1$ serves as the uniquely colored neighbor of $w_2$ or vice versa. W.l.o.g.,
    let $C(w_1) = 1, U(w_1) = 2, C(w_2) = 2, U(w_2) = 3$. Assign $C(w_3) = 1, C(w_4) = 3, C(w_5) = 3$ and $U(w_3) = 2, U(w_4) = 1, U(w_5) = 1$.
    \item $C(w_1) = U(w_2)$ and $C(w_2) = U(w_1)$.
    W.l.o.g.,
    let $C(w_1) = 1, U(w_1) = 2, C(w_2) = 2, U(w_2) = 1$.
    Assign $C(w_3) = 2, C(w_4) = 3, C(w_5) = 1$ and $U(w_3) = 3, U(w_4) = 2, U(w_5) = 3$.
    \item $C(w_1) = C(w_2)$ and $U(w_1) = U(w_2)$.
    W.l.o.g.,
    let $C(w_1) = C(w_2) = 1, U(w_1) = U(w_2) = 2$.
    Assign $C(w_3) = 1, C(w_4) = 2, C(w_5) = 3$ and $U(w_3) = 2, U(w_4) = 3, U(w_5) = 1$.
    \item The case $U(w_1) = U(w_2)$ and $|\{U(w_1), C(w_1), C(w_2)\}| = 3$ does not arise in the above colorings. 
\end{enumerate}
\qed
\end{proof}

At this point, to complete the proof of Theorem \ref{thm:4col}
we need to prove Theorem \ref{thm:otherthan5cycle}. 
We now state 
a few results that would help us towards this end.

\begin{lemma}\label{lem:facecol}
An uncolored face $F$, such that $|V(F)| \neq 5$, can be \cf colored using 4 colors satisfying the invariants.
\end{lemma}
\begin{proof}
Let $F=v_1-v_2-v_3-\dots-v_{k-1}-v_k-v_1$ be a  face with $|V(F)|=k$, $k \neq 5$. 
We assign $C(v_1)=1, C(v_2) = 2, C(v_3) = 3$ and for the remaining vertices (if any), we set $C(v_i)=C(v_{i-3})$. In order
to satisfy the invariants, we need to make the following changes:
\begin{itemize}
    \item $k \equiv 0 \pmod 3$. 
    No change is necessary. 
    \item $k \equiv 1 \pmod 3$. Reassign
     $C(v_k) = 4$.
    \item $k \equiv 2 \pmod 3$. 
     Reassign
     $C(v_{k-3}) = 4, C(v_{k-2}) = 2, C(v_{k-1}) = 3, C(v_k) = 4$.
    Notice that this coloring does not satisfy the invariants if $k = 5$. However, the smallest $k$ that we consider in this case is $k = 8$.
\end{itemize}

In each of the above cases the unique color for each vertex $v_i$ is provided by its cyclical successor i.e., $U(v_i) = C(v_{i+1})$.
\qed
\end{proof}

\begin{lemma}~\label{lem:onevertexcolored}
Let $F$ be a face (cycle) in $G$ with one vertex $v$ such that $C(v)$ and $U(v)$ are already assigned, with $C(v) \neq U(v)$.
Then the rest of $F$ can be complete \cf colored using at most 4 colors, while retaining $C(v)$ and $U(v)$, and 
satisfying the invariants.
\end{lemma}

\begin{proof}
Let $v_1$ be the colored vertex in 
the cycle $F$. 
We may assume w.l.o.g. that 
$C(v_1)=1$ and $U(v_1)=2$. Now, we extend $C$ to the remainder of $F$. 
\begin{itemize}
    \item $|V(F)|=3$ with $F=v_1-v_2-v_3-v_1$. 

    We assign: $C(v_2)=3$, $C(v_3)=4$ and $U(v_2)=1$, $U(v_3)=1$.
    
    
    
    \item $|V(F)|\geq 4$ with $F=v_1-v_2-v_3-\dots-v_{k-1}-v_k-v_1$. 
    
    We first assign: $C(v_2)=3$ and $C(v_{3})=2$. 
    For the remaining vertices $v_i$, we set $C(v_i) = C(v_{i-3})$ 
    for $4 \leq i \leq k$. However, we need to make some changes to this in order 
    to satisfy the invariants. We have the following subcases:
    \begin{itemize}
        \item $k \equiv 0 \mbox{ or } 1 \pmod 3$. 
        Reassign $C(v_k)=4$.
        \item $k \equiv 2 \pmod 3$. 
        Reassign $C(v_{k-1})=4$.
    \end{itemize}
    In each of the above cases the unique color for each vertex $v_i$ is provided by its cyclical successor i.e., $U(v_i) = C(v_{i+1})$. Observe that $U(v_1)$ is left unchanged, by ensuring $v_2$ and $v_k$, the neighbors of
    $v_1$, are not assigned the color $U(v_1)$.
\end{itemize}
\qed
\end{proof}

\begin{lemma}~\label{lem:path}
Let $P$ be a path in $G$ whose endpoints are $v_1, v_2$. Suppose
$\{v_1, v_2\} \in E(G)$ and that $v_1, v_2$ are already assigned 
the functions $C$ and $U$ satisfying the invariants.
Then the rest of $P$ can be \cf colored using at most 4 colors, while retaining $C$ and $U$ values of the endpoints, and satisfying the invariants.
\end{lemma}
Since the proof of the above lemma is a bit long and involved, we first prove Theorem \ref{thm:otherthan5cycle} using Lemmas \ref{lem:facecol}, \ref{lem:onevertexcolored} and \ref{lem:path}.

\begin{proof}[Proof of Theorem \ref{thm:otherthan5cycle}]
\begin{enumerate}
    \item If the block is a bridge, say $\{v, w\}$, then we color it
    $C(v) = 1, C(w) = 2$ with $U(v) = 2, U(w) = 1$. Note that the 
    invariant marked $\star$ is violated in this case. However, this does
    not cause an issue since this edge is a bridge, and it does not appear 
    in any inner face.
    
    \vspace{0.05in}
    If the block is not a bridge, then by assumption, it contains a face $F$ 
    such that $|V(F)| \neq 5$. By Lemma \ref{lem:facecol}, we have a coloring
    of $F$ using 4 colors and satisfying the invariants. By the Lemma \ref{lem:ear} (Ear Decomposition),
    the block has an ear decomposition $F, P_1, P_2, \ldots$ with $F$ as the 
    starting inner face. Recall that for every path $P_i$, the end points form
    an edge in $F_0 \cup \bigcup_{1 \leq j<i} P_j$. 
    We color the paths $P_1, P_2, \ldots$ in this order. 
    By Lemma \ref{lem:path}, we have a coloring for each of these paths 
    using 4 colors and satisfying the invariants. 
    \item Let $v$ be the vertex in the block that is already colored. 
    W.l.o.g., we may assume that 
    $C(v) = 1$ and $U(v) = 2$.
    
    \vspace{0.05in}
    If the block is a bridge $\{v, w\}$,  we color $w$ with $C(w)=3$ and set $U(w) = 1$.
    
    \vspace{0.05in}
    If the block is not a bridge, choose an inner face $F$ that contains $v$.
    Using Lemma \ref{lem:onevertexcolored}, we color the remainder of $F$ using 
    at most 4 colors and satisfying the invariants. The rest of the proof 
    follows from the fact that we have an ear decomposition with $F$ as the 
    starting face, and Lemma \ref{lem:path}.
    This is very similar to the argument in the proof of part 1 of this theorem
    and hence the details are omitted.
\end{enumerate}
\vspace{-0.1in}
\qed
\end{proof}



\noindent In order to complete the proof of Theorem \ref{thm:4col}, the last remaining 
piece is the proof of Lemma \ref{lem:path}. 

\begin{proof}[Proof of Lemma~\ref{lem:path}] Let $v_1$ and $v_2$ be the end points of $P$. We extend the coloring $C$ to the remainder of $P$. 
According to the invariants of $C$, we have only 2 cases possible.

\noindent\textbf{Case 1:} $C(v_1) \neq C(v_2), U(v_1) \neq U(v_2)$. W.l.o.g. we may assume  
$C(v_1)=1$, $C(v_2)=2$ and $U(v_1)=2$, $U(v_2)=3$.
\begin{itemize}
    \item  $|V(P)|=3$, $P=v_2-v_3-v_1$. Assign $C(v_3)=4$ with $U(v_3)=2$.
    \item  $|V(P)|=4$, $P=v_2-v_3-v_4-v_1$. Assign $C(v_3)=4$, $C(v_4)=3$ with $U(v_3)=3$, $U(v_4)=1$. 

    \item  $|V(P)|\geq5$, $P=v_2-v_3-\dots -v_{k-1}-v_k-v_1$. 
    We first assign $C(v_3)=1, C(v_4) = 3, C(v_5) =4$. 
    For the remaining vertices $v_i$, we initially assign $C(v_i) = C(v_{i-3})$ 
    for $6 \leq i \leq k$. However, we need to make some changes 
    to satisfy the invariants. We have the following subcases:
    \begin{itemize}
        \item $k \equiv 0 \pmod 3$. 
        Reassign $C(v_{k-1})=2$ and $C(v_k) = 4$
        \item $k \equiv 1 \pmod 3$. 
        Reassign $C(v_{k-1})=2$.
        \item $k \equiv 2 \pmod 3$. 
        No change is necessary.
    \end{itemize}
    In each of the above cases the unique color for each vertex $v_i$ is provided by its cyclical successor i.e., $U(v_i) = C(v_{i+1})$.
\end{itemize}

\noindent \textbf{Case 2:} $U(v_1) = U(v_2)$. W.l.o.g., we may assume $C(v_1)=1$, $C(v_2)=2$ and $U(v_1)= U(v_2)=3$. 

\begin{itemize}
    \item  \textbf{Case 2(i):} $|V(P)|=3$ and $P=v_2-v_3-v_1$.
    \begin{itemize}
        \item \textbf{Case 2(i)(a):} Vertices $v_1$ and $v_2$ are the only
        neighbors of $v_3$. Assign $C(v_3)=4$ and $U(v_3)=2$. The  
        invariant marked $\star$ is not satisfied, but that does not matter as 
        $v_3$ does not participate in any further faces.
        
        \item \textbf{Case 2(i)(b):} One of the edges $\{v_1, v_3\}$ or $\{v_2, v_3\}$ does
        not feature in an another face. W.l.o.g., say $\{v_2, v_3\}$ be that edge. 
        Assign $C(v_3) = 4$ with $U(v_3) = 1$. The $\star$ invariant is violated
        for $\{v_2, v_3\}$ here but it does not affect the further coloring.
        
        \item \textbf{Case 2(i)(c):} One of the edges $\{v_1, v_3\}$ or $\{v_2, v_3\}$ 
         features in an uncolored face $F$ such that $|V(F)| \neq 3$. 
         W.l.o.g., say $\{v_2, v_3\}$ is that edge.

         We assign $C(v_3) = 4$ with $U(v_3) = 1$. Let $|V(F)| = k$ with 
         $F=v_3-w_1-w_2-\dots-w_{k-2}-v_2-v_3$.
         We assign $C(w_1)=3$, $C(w_2)=1$ and $C(w_3)=4$ (if $w_3$ exists). 
         For all $4\leq i\leq k-2$, $C(w_i)=C(w_{i-3})$. 
         If $k \equiv 0 \pmod 3$, we reassign $C(w_{k-4})=2$, $C(w_{k-3}) = 1$ and $C(w_{k-2})=4$.
         
         The unique colors $U$ for the vertices
         are assigned as follows:
         \begin{itemize}
             \item For $k=6$, $U(w_1) = 4, U(w_2) = 3, U(w_3) = 2$ and
             $U(w_4) = 2$.
             \item For $k\neq 6$, we have 
             for $1\leq i \leq k-3$, 
         $U(w_i) = C(w_{i+1})$ and $U(w_{k-2})  = C(v_2) = 2$. 
         \end{itemize}


    
    
    

         
\begin{figure}[t!]
\centering
\begin{tikzpicture}[every node/.style={node distance=1.8cm,scale=0.9}, scale = 0.8]
\tikzstyle{vertex}=[circle,draw, minimum size=8pt]
\tikzstyle{edge} = [draw,thick,-,black]
\node[vertex]  (v3) at (0,-2) {$v_3$};
\node[vertex]  (v1) at (-1.2,2) {$v_1$};
\node[vertex]  (v2) at (1.2,2) {$v_2$};
\node[vertex]  (x) at (-3,0) {$x$};
\node[vertex]  (y) at (3,0) {$y$};
\node[vertex]  (z) at (-2.2,-2.5) {$z$};

\draw[edge, color=black] (x) -- (v1) -- (v2) -- (y);
\draw[edge, color=black] (x) -- (v3) -- (v1);
\draw[edge, color=black] (y) -- (v3) -- (v2);
\draw[edge, color=black] (x) -- (z) -- (v3);
\end{tikzpicture}
\caption{Case 2(i)(d)}\label{fig:2id}
\end{figure}
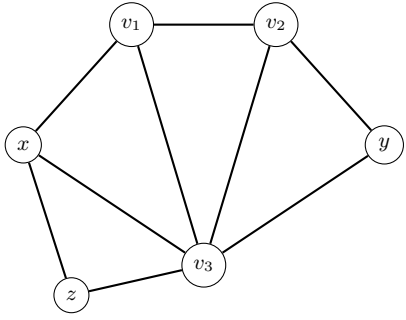

        \item \textbf{Case 2(i)(d):} The only remaining case is when both
         the edges $\{v_1, v_3\}$ or $\{v_2, v_3\}$ feature in uncolored triangular
         faces. Let $\{v_1, v_3\}$ form a triangular face with $x$ and 
         $\{v_2, v_3\}$ with $y$. We have two subcases:
        \begin{itemize}
            \item The edge $\{x, v_3\}$ forms a triangular face with another 
            vertex $z$ (see Figure \ref{fig:2id}).  
            Assign $C(v_3) = 1, C(x) = 2, C(y) = 4, C(z) = 3$ and $U(v_3) = 4, U(x) = 3, 
            U(y) = 2, U(z) = 1$. Some edges 
            violate the invariant marked $\star$, but these edges are already
            part of two faces, and hence do not feature in the further coloring.
            
            \item The edge $\{x, v_3\}$ is not part of a triangular face with another
            vertex. In this case, we assign $C(v_3) = 4, C(x) = 4, C(y) = 1$ and $U(v_3) = 2, U(x) = 1, U(y) = 2$. Out of the edges that violate 
            the invariant marked $\star$, the only one that can participate in the
            further coloring is the edge $\{x, v_3\}$. By assumption, $\{x, v_3\}$
            is not part of a triangular face. In Lemma \ref{lem:CuCvsame}, 
            we explain how to color the uncolored face that is $\{x, v_3\}$ may be 
            a part of.
        \end{itemize}
    \end{itemize}
    
    \item \textbf{Case 2(ii):} $|V(P)|=4$, $P=v_2-v_3-v_4-v_1$. 
    \begin{itemize}
        \item \textbf{Case 2(ii)(a):} The edge $\{v_3, v_4\}$ forms a triangular
        face with a vertex $x$.
        We assign $C(v_3)= 1, C(v_4) = 4, C(x) = 3$, with $U(v_3)= 3, U(v_4) = 1, U(x) = 4$.
        \item \textbf{Case 2(ii)(b):} The edge $\{v_3, v_4\}$ is not part of 
        an uncolored triangular face. We assign $C(v_3) = C(v_4) = 4$, with 
        $U(v_3)= 2, U(v_4) = 1$. If the edge  $\{v_3, v_4\}$ is part of an
        uncolored face $F$, by assumption, we know that $|V(F)| \geq 4$ and hence
        we can use Lemma \ref{lem:CuCvsame} to color $F$ satisfying the invariants.
    \end{itemize}
    \item \textbf{Case 2(iii):} $|V(P)|= 5$ with $P = v_2-v_3-v_4-v_5-v_1$. We
    assign $C(v_3)=1, C(v_4) = 3, C(v_5) =2$, with $U(v_3)=3, U(v_4) = 2, U(v_5) =1$.
    \item \textbf{Case 2(iv):} $|V(P)|\geq 6$, with $P=v_2-v_3-\dots-v_{k-2}-v_{k-1}-v_k-v_1$.

    We first assign $C(v_3)=4$ and $C(v_4)=3$. 
    For $5\leq i \leq k$, assign $C(v_i)=C(v_{i-3})$. 
    If $k \equiv 1 \pmod 3$, then reassign $C(v_{k-2})=1$ and $C(v_k)=2$. 
    For each vertex $v_i$, the unique color is provided by its cyclical successor
    i.e., $U(v_i) = C(v_{i+1})$.
\end{itemize} 
\qed
\end{proof}

\begin{lemma}\label{lem:CuCvsame}
Let $F$ be a face with $|V(F)|\geq 4$ with such that the edge $\{v_1, v_2\} \in E(F)$ 
and $v_1$ and $v_2$ already colored such that 
$C(v_1)=C(v_2)$ and $U(v_1)\neq U(v_2)$.
Then the rest of $F$ can be \cf colored using 4 colors satisfying the invariants. 
\end{lemma}

\begin{proof}
W.l.o.g., we may assume $C(v_1)=C(v_2)=4$, $U(v_1)=1$ and $U(v_2)=2$. We have the following
cases:
\begin{itemize}
    \item $|V(F)|=4$ with $F=v_2-v_3-v_4-v_1-v_2$. 
    We assign: $C(v_3)=1$, $C(v_4)=3$ and $U(v_3)=4$ and $U(v_4)=4$.
     \item If $|V(F)|= 5$ with $F=v_2-v_3-v_4-v_5-v_1-v_2$.
     We assign: $C(v_3)=1$, $C(v_4)=2$, $C(v_5)=3$ and $U(v_3)=2$, $U(v_4)=3$ and $U(v_4)=4$.
    \item If $|V(F)|\geq 6$ with $F=v_2-v_3-\dots- v_{k-1} - v_{k}-v_1 -v_2$. 
    We assign: $C(v_3)=3$ and $C(v_4)=2$. For all $5\leq i \leq k$, $C(v_i)=C(v_{i-3})$.
    \begin{itemize}
        \item $k \equiv 0 \pmod 3$. Reassign $C(v_{k-1})=1$. 
        \item $k \equiv 1 \pmod 3$. No change is required.
         \item $k \equiv 2 \pmod 3$. Reassign $C(v_{k-1})=1$ and $C(v_k)=2$.
    \end{itemize}
    The unique color of each vertex $v_i$ is provided its cyclical successor i.e., $U(v_i) = C(v_{i+1})$.
\end{itemize}
\qed
\end{proof}

\noindent\textbf{Algorithmic Note:} The steps in the proof of Theorem \ref{thm:4col}
leads to an algorithm. Block decomposition, outerplanarity testing and embedding
outerplanar graphs \cite{outerplanar} can all be done in linear time, i.e., $O(|V(G)|)$. Thus we have 
an $O(|V(G)|)$ time algorithm, that given an outerplanar graph $G$, determines a complete
\cf coloring for $G$ that uses four colors.

\subsection{Cactus Graphs}

Now, we show that a cactus graph can be complete 
\cf colored using 3 colors. This is a tight bound.

\begin{definition}\label{def:cactus}
A cactus graph is a connected graph in which any two cycles have at most one vertex in common.
\end{definition}

\begin{theorem}
Three colors are sufficient and sometimes necessary to complete 
\cf 
color cactus graphs.
\end{theorem}
\begin{proof}
Cactus graphs are outerplanar and by  definition~\ref{def:cactus}, 
any two cycles have at most one vertex in common. 
We apply the block decomposition on the cactus graph $G$. 
Note that each block is a cycle or a bridge. 
Throughout the coloring, we maintain the invariant that for each vertex $v$, 
the unique color seen by $v$,  $U(v) \neq C(v)$.

Let $B_0$ be the first block considered. 
If $B_0$ is a bridge $\{v, w\}$,  we assign $C(v)=1$, $C(w)=2$ and $U(v) = 2$, $U(w)=1$. 
Else $B_0$ is a cycle with  $F=v_1-v_2-\dots-v_{k-1}-v_{k}-v_1$. 
Initially we assign $C(v_1)=1$, $C(v_2)=2$ and $C(v_3)=3$. 
For all $4\leq i \leq k$, $C(v_i)=C(v_{i-3})$. If $k \equiv 2 \pmod 3$, $C(v_{k-1})=3$.
In each of the cases, for each vertex $v$, we can identify $U(v)$ such that $U(v) \neq C(v)$.

Now, we choose a block that is adjacent to an already colored block. 
Such a block has exactly one colored vertex. Let $v_1$ 
be the that vertex with $C(v_1) = 1$ and $U(v_1) =2$.
If the block is a bridge, say $\{v_1, w\}$, then assign $C(w) = 3$ with $U(w) = 1$.
Else the block is an inner face $F=v_1-v_2-\dots-v_{k-1}-v_{k}-v_1$ 
with a colored vertex $v_1$. 
We initially assign $C(v_2)=3$, $C(v_3)=2$. For all $4\leq i \leq k$, $C(v_i)=C(v_{i-3})$. 
In the case when $k \equiv 0 \pmod 3$, we reassign $C(v_k) = 3$, and in the case when 
$k \equiv 2 \pmod 3$, we reassign $C(v_{k-1}) = 2$. In this case too, we maintain the
invariant that $C(v) \neq U(v)$ for each $v$.

To see that the bound is tight, observe that Figure~\ref{fig:k3} is a cactus graph that requires three colors.
\qed
\end{proof}

%% file: kneser.tex
\section{Kneser graphs}\label{sec:kneser}

In this section, we study the 
\cf coloring of Kneser graphs. 
Throughout this section, we use $[n]$ to denote the set $\{1, 2, .., n\}$. 
The Kneser graph $K(n,k)$ is formally defined as follows: 
\begin{definition}[Kneser graph]
The Kneser graph $K(n,k)$ is the graph whose vertices are $\binom{[n]}{k}$, the $k$-sized subsets of $[n]$, and the vertices $x$ and $y$ are adjacent if and only if $x\cap y = \emptyset$ (when $x$ and $y$ are viewed as sets).
\end{definition}

During the discussion, we shall use the words \emph{$k$-set} or \emph{$k$-subset} to refer 
to a set of size $k$. We shall sometimes refer to the $k$-subsets of $[n]$ and the vertices
of $K(n,k)$ in an interchangeable manner.

\begin{lemma}\label{lem:cfonk+1}
 $k+2$ colors are sufficient to complete \cf 
 color a $K(n,k)$ when $n\geq 3k-1$. 
\end{lemma}

\begin{proof}
We first assign 
a coloring\footnote{In this coloring, 
the uniquely colored neighbor is not colored $k+2$ for any 
of the vertices. Thus, by recoloring the color class $k+2$
with the color 0, we get a partial coloring that uses $k+1$ colors.}
to the vertices of 
$K(n,k)$ and then argue 
that this coloring is a complete \cf 
coloring. 
\begin{itemize}
    \item For any vertex ($k$-set) $v$ that is a subset of $\{1, 2, \dots, 2k-1\}$, we assign $C(v) = \max_{\ell \in v} \ell - (k-1)$.
    \item The set $\{2k, 2k+1, \dots, 3k-1\}$ is assigned the color $k+1$.
    \item All the remaining vertices are assigned the color $k+2$.
\end{itemize}
For example, for the Kneser graph $K(n,3)$, 
we assign the color 1 to the vertex $\{1, 2, 3\}$, 
color 2 to the vertices $\{1, 2, 4\}, \{1, 3, 4\}, \{2, 3, 4\}$, 
color 3 to the vertices $\{1, 2, 5\}, \{1, 3, 5\}, \{1, 4, 5\}, \{2, 3, 5\}, \{2, 4, 5\}, \{3, 4, 5\}$,
color 4 to the vertex $\{6, 7, 8\}$ and 
color 5 to all the remaining vertices.

Now, we prove that in the above coloring, every vertex
has a uniquely colored neighbor.
Let $C_i$ be the set of all vertices assigned the color $i$, i.e., 
the color class of the color $i$.
Notice that $C_1 \cup C_2 \cup \dots \cup C_k = \binom{[2k-1]}{k}$.
Let $w_{k+1}$ denote the $k$-set $\{2k, 2k+1, \dots, 3k-1\}$.
Any vertex $v \in C_1 \cup C_2 \cup \dots \cup C_k$ is a neighbor of $w_{k+1}$.
Since $w_{k+1}$ is the lone vertex colored $k+1$, it serves as the uniquely colored
neighbor for any $v \in C_1 \cup C_2 \cup \dots \cup C_k$.

Now we have to show the presence of uniquely colored neighbors for 
vertices that have some elements from outside $\{1, 2, \dots, 2k-1\}$. 
Let $v$ be the vertex such that it has some elements from outside 
$\{1, 2, \dots, 2k - 1\}$. That is, $v \cap \{1, 2, \dots, 2k - 1\} \neq v$. Let $t = t(v)$
be the smallest nonnegative integer such that $\left | \{1, 2, \dots, k + t\} \setminus v \right | = k$. Since $v$ has at least one element 
from outside $\{1, 2, \dots, 2k-1\}$, $t$ is at most $k-1$.

We claim that $v$ has a lone neighbor colored $t+1$, and this neighbor is given 
by the set $\{1, 2, \dots, k + t\} \setminus v$. By the choice of $t$, 
this is the only neighbor of $v$ that is colored $t+1$. It can be further
observed that there are no neighbors of $v$ that are assigned a  
color smaller than $t+1$.
\qed
\end{proof}

Now we show that $k+2$ colors are necessary, when $n$ is large enough.
\begin{theorem}\label{lowerbound}
$k+2$ colors are necessary to complete \cf 
color 
$K(n,k)$ when $n\geq k(k+1)^2 + 1$.
\end{theorem}

\begin{proof}
We prove this by contradiction. Suppose that $K(n,k)$ can be colored using the $k+1$
colors $1, 2, 3, \dots, k+1$. 
For each $1 \leq i \leq k+1$, let $C_i$ denote the 
color class corresponding to the color $i$, i.e., the set of all vertices colored with 
the color $i$. Let $q = \binom{n}{k}/(k+1)$. 

If for all $i$, $|C_i| < q$, this implies that the total number of vertices is strictly
less than $q (k+1) = \binom{n}{k}$. This is a contradiction. Hence there is at least one 
$i$, such that $|C_i| \geq q$.
For any vertex $v$, let $d_i(v)$ denote the number of neighbors of $v$ in $C_i$. 

Our strategy is to find one vertex, say $x$, 
which does not have a uniquely colored neighbor. More formally, 
we want $x$ to satisfy $d_i(x) \neq 1$, for all $1 \leq i \leq k+1$.
We construct the vertex ($k$-set) $x$, by choosing elements in it as follows. Suppose there
are $C_i$'s that are singleton, i.e., $|C_i| = 1$. For all the singleton $C_i$'s we choose
a hitting set. In other words, we choose entries in $x$ so as to ensure that $x$ intersects
with the vertices in all the singleton $C_i$'s. This partially constructed $x$ may also 
 intersect with vertices in other $C_i$'s. Some of the other $C_i$'s might become 
 ``effectively singleton'', that is $x$ may intersect with all the vertices in those $C_i$'s 
 except one. We now choose further entries in $x$ so as to hit these effectively singleton
 $C_i$'s too. Finally, we terminate this process when all the   
 remaining $C_i$'s are not singleton. 
 
 At this stage, $x$ can have potentially $k+1$ entries, one each 
 to hit the $k+1$ color classes. However, the below claim shows that
 not all the color classes need to be hit. 
 

\noindent \textbf{Claim:} There exists an $i$ for which $C_i$ does not become singleton/effectively singleton.

\noindent \emph {Proof of claim.} We have already seen that there is at least one $C_i$
for which $|C_i| \geq q = \frac{\binom{n}{k}}{k+1}$. We show that this $C_i$ does not
become effectively singleton. 

Let $t$ be the number of entries in $x$ when the 
above process terminates. Notice that each entry in $x$ can cause $x$ to intersect
with at most $\binom{n-1}{k-1}$ other vertices. We have $t\leq k+1$ entries in $x$, so $x$ can intersect with at most 
$(k+1) \binom{n-1}{k-1}$ vertices. 
When $n \geq k(k+1)^2 + 1$, it can be
verified that $(k+1) \binom{n-1}{k-1} < q - 2$, leaving at least two vertices in $C_i$ that
do not intersect with $x$.
\qed

Due to the above claim, the number of entries in $x$ is $t \leq k$.
To fill up the remaining entries of $x$ (if any), we consider the set(s) $C_j$ that have not become effectively singleton. For each of these sets $C_j$, we choose
two distinct vertices, say $y_j, y_j' \in C_j$. We choose the remaining entries of $x$ so that $x \cap y_j = \emptyset$ and $x \cap y_j' = \emptyset$. The number of 
such sets $C_j$ is at most $k+1$. So for choosing the remaining entries of $x$, we have
at least $n -t - 2k(k+1)$ choices. 
Because $n > k^3$, we can choose such entries.
\qed
\end{proof}

It is worth noting that the above proof technique 
cannot be applied for showing a lower bound of $k+3$. For such a proof, we would start with a
$k+2$ coloring, and try for a contradiction. In this case, we could have $k+1$ singletons and effective singletons, which could require
$k+1$ elements of $[n]$ to hit. However, $x$
can hold at most $k$ elements. This is where the
proof breaks down.

\subsection{\cf Closed Neighborhood 
Coloring of Kneser Graphs}
In this section, we see some results on the 
\cf closed neighborhood 
coloring 
of Kneser graphs.
We abbreviate 
\cf closed neighborhood 
coloring as \cfcn coloring and denote by $\chi_{CF-CN}(G)$, 
the conflict-free closed neighborhood chromatic number of $G$.
It is easy to see that a proper 
coloring of a graph $G$ is also a 
\cfcn 
coloring. That is, $\chi_{CF-CN}(G) \leq \chi(G)$ for all 
graphs $G$.
Since $\chi(K(n,k)) \leq n - 2k + 2$ \cite{lovasz}, we have that
$\chi_{CF-CN}(K(n,k)) \leq n - 2k + 2$.

\begin{lemma}\label{cfcnk+1} 
When $n\geq 2k+1$, we have $\chi_{CF-CN}(K(n,k)) \leq k+1$.
\end{lemma}

\begin{proof}
We assign the following coloring to the vertices of $K(n,k)$: 


\begin{itemize}
    \item For any vertex ($k$-set) $v$ that is a subset of $\{1, 2, \dots, 2k-1\}$, we assign $C(v) = \max_{\ell \in v} \ell - (k-1)$.
    \item All the uncolored vertices are assigned color $k+1$.
\end{itemize}


For $1 \leq i \leq k+1$, let $C_i$ 
be the color class of the color $i$. 
Notice that $C_1 \cup C_2 \cup \dots \cup C_k = \binom{[2k-1]}{k}$. Since any two $k$-subsets of $\{1, 2, \dots, 2k-1\}$ intersect,
it follows that 
$\binom{[2k-1]}{k}$ is an independent set. Hence each of the color classes $C_1, C_2, \dots, C_k$ are independent sets.
So if $v$ is colored with color $i$, where $1 \leq i \leq k$, it has no neighbors of its own color. Hence, it serves as its own uniquely colored neighbor.

If $v$ is colored $k+1$, then $v \not \subset [2k-1]$. That is, $v$ has some elements from outside $[2k-1] = \{1, 2, \dots, 2k - 1\}$. 
Let $t = t(v)$
be the smallest nonnegative integer such that $\left | \{1, 2, \dots, k + t\} \setminus v \right | = k$. Since $v$ has at least one element 
from outside $\{1, 2, \dots, 2k-1\}$, $t$ is at most $k-1$.
It is easy to verify that the vertex corresponding to the set $\{1, 2, \dots, k + t\} \setminus v$ is the lone 
neighbor of $v$ that is colored $t+1$, and thus serves as the uniquely colored neighbor of $v$.
\qed
\end{proof}

\begin{lemma}\label{2k+1lemma}
 $\chi_{CF-CN}(K(2k+1,k)) = 2$, for all $k\geq 1$.
\end{lemma}
\begin{proof}
Consider a vertex $v \in V(K(2k+1, k))$.
If $v\cap \{1,2\} \neq \emptyset$, we assign color 1 to $v$. Else, we assign color 2 to $v$.

Let $C_1$ and $C_2$ be the sets of vertices colored $1$ and $2$ respectively. Below, we 
discuss the unique colors for every vertex of $K(n,k)$.
\begin{itemize}
    \item If $v\in C_1$ and $\{1, 2\} \subseteq v$, then $v$ is the uniquely colored neighbor of itself. This is because all the
    vertices in $C_1$ contain either 1 or 2 and hence $v$ has no neighbors in $C_1$.
    
    \item Let $v\in C_1$ and $|v \cap \{1, 2\}| = 1$. W.l.o.g., let $1 \in v$ and $2 \notin v$.
    In this case, $v$ has a uniquely colored neighbor $w\in C_2$. The vertex $w$ is the
    $k$-set $w = [2k +1] \setminus (v \cup \{2\})$.
    
    \item If $w\in C_2$, $w$ is the unique color neighbor of itself. This is because 
    $C_2$ is an independent set. For two vertices 
    $w,w'\in C_2$ to be adjacent, we need $|w \cup w'| = 2k$, 
    but vertices in $C_2$ are subsets of $\{3, 4, 5, \ldots, 2k+1\}$, which has
    cardinality $2k-1$.
\end{itemize}
\qed
\end{proof}

\begin{lemma}\label{cfcnd+1}
$\chi_{CF-CN}(K(2k + d,k)) \leq d+1$, for all $k\geq 1$. 
\end{lemma}
\begin{proof}
We prove this by induction on $d$. 
The base case is when $d=1$ which is true from Lemma~\ref{2k+1lemma}. 
Suppose $K(2k+d, k)$ has a \cfcn coloring with $d+1$ colors.
Let us consider $K(2k+d+1, k)$. For all the vertices
of $K(2k+d+1, k)$ that appear in $K(2k+d, k)$
we use the same coloring as in $K(2k+d, k)$.
The new vertices (the vertices that contain 
$2k + d +1$) are assigned the new color $d+2$.
As all the new vertices contain $2k+d+1$, they 
form an independent set. Hence each of the new 
vertices serve as their own uniquely colored
neighbor. 

The vertices of $K(2k+d+1, k)$ already present in $K(2k+d, k)$ get new neighbors, but all the new neighbors 
are colored with the new color $d+2$. Hence the 
 unique color of the existing vertices
are retained.
\qed
\end{proof}

So, from Lemma~\ref{cfcnk+1} and Lemma~\ref{cfcnd+1} we get the following. 

  \[
    \chi_{CF-CN}[K(n,k)]  \leq \left\{\begin{array}{lr}
        n-2k+1, & \text{for } 2k+1\leq n\leq 3k 
        \\
        k+1, & \text{for } n\geq 3k+1\\
        \end{array}\right\}.
  \]

%% file: conclusion.tex


\section{Discussion}
We note a few directions that are left open by this paper:
\begin{itemize}
    \item We showed that a planar graph has a partial \cf coloring
    that uses at most 5 colors. The best known lower bound is 4 colors.
    \item Along similar lines, an outerplanar graph can be partial
    \cf colored using 4 colors, while the lower bound is 3.
    \item We showed that the complete \cf chromatic number of $K(n,k)$
    is $k+2$ when $n \geq k(k+1)^2 + 1$. We believe this requirement on $n$ 
    can be relaxed.
\end{itemize}

\noindent \textbf{Acknowledgments:} We would like to thank I. Vinod Reddy for suggesting the problem, 
Karthik R. for initial discussions and 
Rogers Mathew for the proof of Theorem~\ref{lowerbound}.